\documentclass[11pt]{article}
\usepackage[T1]{fontenc}
\usepackage[utf8]{inputenc}
\usepackage{graphicx}
\usepackage{amsmath} 
\usepackage{amssymb}
\usepackage{authblk}
\usepackage{url}
\usepackage{enumerate}
\usepackage{color}
\usepackage{epsfig}
\usepackage{xspace}

\usepackage{algorithmic}
\usepackage{algorithm}
\usepackage{enumitem}
\newtheorem{observation}{Observation}

\newcommand{\procSingleTest}{\texttt{\textup{TestPorts}}}

\newcommand{\procBoundDegrees}{\texttt{\textup{BoundDegrees}}}
\newcommand{\procCompareLabels}{\texttt{\textup{CompareLabels}}}
\newcommand{\procRV}{\texttt{\textup{Rendezvous}}}
\newcommand{\returnSuccess}{\texttt{\textup{success}}}
\newcommand{\returnFailure}{\texttt{\textup{failure}}}

\newcommand{\moveRound}[1]{\texttt{\textup{move}}(#1)}
\newcommand{\currentDist}{\texttt{\textup{distance}}}
\newcommand{\agent}{A}
\newcommand{\dist}[2]{\textup{dist}(#1,#2)}

\newcommand{\degree}[1]{d(#1)}
\newcommand{\labelOfAgent}[1][]{\if!#1!\ell\else\ell(#1)\fi}
\newcommand{\labelExtended}[1]{\xi(#1)}
\newcommand{\labelsBit}[2]{#1_{#2}}      

\newcommand{\qed}{\hspace*{\fill}\rule{6pt}{6pt}\vspace{.5\smallskipamount}}
\newtheorem{theorem}{Theorem}[section]
\newtheorem{lemma}{Lemma}[section]

\newtheorem{definition}{Definition}[section]

\newenvironment{proof} { \noindent \emph{Proof}: } { \qed }

\begin{document}

\title{Rendezvous of Distance-aware Mobile Agents \\in Unknown Graphs\thanks{This work was done when the second and the third author were visiting the LIF laboratory in Marseille.
Research partially supported by the Polish National Science Center grant DEC-2011/02/A/ST6/00201, by the ANR project DISPLEXITY (ANR-11-BS02-014) and by the ANR project MACARON (ANR-13-JS02-0002-01).
Dariusz Dereniowski was partially supported by a scholarship for outstanding young researchers funded by the Polish Ministry of Science and Higher Education.}}

\author[1]{Shantanu Das}
\author[2]{Dariusz Dereniowski}
\author[3,4]{Adrian Kosowski}
\author[1]{Przemys\l{}aw Uzna\'{n}ski}

\affil[1]{LIF, Aix Marseille University and CNRS, Marseille, France.}
\affil[2]{Department of Algorithms and System Modeling, Gda\'{n}sk University of Technology, Gda\'{n}sk, Poland.}
\affil[3]{GANG Project, Inria Paris, France.}
\affil[4]{LIAFA, Paris Diderot University and CNRS, France.}

\date{}
\maketitle
\thispagestyle{empty}

\begin{abstract}
We study the problem of rendezvous of two mobile agents starting at distinct locations in an unknown graph.
The agents have distinct labels and walk in synchronous steps. However the graph is unlabelled and the agents have no means of marking the nodes of the graph and cannot communicate with or see each other until they meet at a node.
When the graph is very large we want the time to rendezvous to be independent of the graph size and to depend only on the initial distance between the agents and some local parameters such as the degree of the vertices, and the size of the agent's label. It is well known that even for simple graphs of degree $\Delta$, the rendezvous time can be exponential in $\Delta$ in the worst case.
In this paper, we introduce a new version of the rendezvous problem where the agents are equipped with a device that measures its distance to the other agent after every step. We show that these \emph{distance-aware} agents are able to rendezvous in any unknown graph, in time polynomial in all the local parameters such the degree of the nodes, the initial distance $D$ and the size of the smaller of the two agent labels $l = \min(l_1, l_2)$.  Our algorithm has a time complexity of $O(\Delta(D+\log{l}))$ and we show an almost matching lower bound of $\Omega(\Delta(D+\log{l}/\log{\Delta}))$ on the time complexity of any rendezvous algorithm in our scenario. Further, this lower bound extends existing lower bounds for the general rendezvous problem without distance awareness. \ \\


\end{abstract}

\newpage
\pagenumbering{arabic}

\section{Introduction}

\subsection{Overview}
Suppose two friends travel to a distant land and arrive at a city where all road signs are written in a language unknown to either of them. If the friends get separated and can no longer communicate with each-other, how could the two friends get together again without any help from a third person. This problem of gathering two autonomous mobile agents, called the \emph{rendezvous} problem has been studied in many different contexts, for example for two ships lost in the sea, two astronauts that land in separate parts of a planet and so on. Initial studies on the problem were restricted to finding probabilistic strategies for movement of the two agents that minimize the expected time to rendezvous (See \cite{AlpG03} for a survey of such results). In recent years, the deterministic version of the problem has received a lot of attention especially by the distributed computing community. The rendezvous problem for two agents moving along the edges of a graph is a typical problem of symmetry breaking and it is a primitive for distributed coordination among autonomous mobile robots. The solution to the problem depends on the structure of the graph, the capabilities of the agents and the initial knowledge available to the agents. In this paper, we consider the problem for deterministic agents with local vision moving in an initially unknown graph; the problem is solved when the two agents are simultaneously located in the same vertex. We are interested in the worst case time complexity for rendezvous.
In general, rendezvous cannot always be solved deterministically when the underlying graph is highly symmetric and the agents follow identical strategies. A typical example is that of a ring network with unlabelled nodes where the two agents are placed on opposite vertices on any diameter. In this case, the distance between the two agents may never decrease if the agents use identical strategies, moving left or right at the same time.

Known solutions to the rendezvous problem are based on one of the following two approaches. The first type of solutions relies on finding a point of asymmetry in the graph and meeting at a unique point of asymmetry (e.g. such a point of asymmetry always exists in graphs where the nodes are labelled uniquely). The second type of solutions assumes that the agents are provided with distinct labels and thus, they can execute distinct strategies and ensure rendezvous. The former type of results require the agents to traverse every edge of the graph in the worst case and the time to rendezvous depends on the size of the graph. On the other hand, the latter type of solutions allow the agents to rendezvous in graphs of arbitrary size or even infinite graphs when the agents are located a finite distance apart. It has been shown that rendezvous of agents with distinct labels can be achieved in arbitrary finite graphs in time polynomial in the size of the graph and the size of the smaller of the labels assigned to the agents both in the synchronous case~\cite{KowalskiM08,TSZ07} and the asynchronous case~\cite{DieudonnePV2013}. For infinite graphs, the only known results are for very specific graphs such as lines~\cite{Stachowiak} or grids~\cite{BampasGrid2010}. For a $\Delta$-dimensional infinite grid, the optimal time to rendezvous is $\Theta(D^\Delta)$ which is already exponential in the maximum degree of the graph. In fact, in unknown graphs of degree $\Delta$, where the agents start a distance of $D$ apart, an agent may have to visit all vertices at a distance of $D$ from its initial location.  Since there could be $D^{\Delta}$ such vertices, the time cost of rendezvous would be exponential in $\Delta$, even if the agents have complete knowledge of the graph as well as the initial distance between them. Thus the question is what additional capabilities would enable the agents to rendezvous in polynomial time. 

In this paper we are interested in designing the simplest mechanism that can help the agents to rendezvous in a large (possibly infinite) graph in time polynomial in the other parameters of the problem, e.g. the initial distance, the maximum degree of the graph and the labels assigned to the agents. We achieve this by equipping the agents with a device that can measure the distance to the other agent in the graph after each step of the algorithm\footnote{Such a device can be implemented in practice e.g. by emitting specific signals at periodic intervals and measuring the intensity of the signal received from the other device held by the other agent.}. In fact, our algorithm does not require knowledge of the exact distance between the agents, but instead it is sufficient if the agent can detect whether the distance to the other agent increased or decreased after each move. We assume time to be discretized into rounds; in each round an agent can either traverse one edge of the graph or stay at its current location, and at the end of each round the agent can determine whether the distance to the other agent increased, decreased or remained unchanged during this round. Note that the agents have no means of detecting the direction leading to the other agent. We call agents equipped with the above device \emph{distance-aware} agents. We show that distance-aware agents can rendezvous in arbitrary graphs in time polynomial in the initial distance, in the degree of the graph, and in the size of the smaller of the two agent labels.

\subsection{Our Contributions}

We show that two distance-aware agents, starting from an initial distance $D$ apart from each other in a connected graph with maximum degree $\Delta$, can rendezvous in time $O(\Delta \cdot D + \Delta \cdot \log{l})$ rounds, where $l = \min(l_1, l_2)$ and $l_1$ and $l_2$ are the labels of the agents. The proof is constructive and provides a deterministic algorithm for the agent that takes as input the label of the agent. The algorithm does not require any prior knowledge of the graph and works for any connected graph. We also show that our algorithm is almost optimal by providing a lower bound of $\Omega(\Delta (D+\log{l}/\log{\Delta}))$ for rendezvous of distance-aware agents. Thus, our algorithm is asymptotically optimal when the maximum degree $\Delta$ is not extremely large. 

The lower bound presented in this paper holds even for agents that can compute the exact distance to each other at every step, while the algorithm requires only the knowledge of changes in distance. Moreover, this lower bound extends existing lower bounds for the general problem of rendezvous of labelled agents in unknown graphs. In terms of the size of agent labels, the previous lower bound (without distance awareness) was $\Omega(\log{l}\cdot D)$ which already holds for the ring; no generalizations of this lower bound to graphs of arbitrary degree have been presented before. Our results show that the lower bound for rendezvous must be at least $\Omega(\log{l}(D+\Delta/\log{\Delta}))$. 

\subsection{Related Work}

This paper considers the deterministic version of rendezvous (for randomized solutions see e.g. \cite{AlpG03}).
The problem of rendezvous of autonomous mobile agents has been studied for agents moving in a discrete space i.e. a graph~\cite{PelcSurvey} or those moving on a continuous space (e.g. two dimensional plane \cite{robot05}). In the graph setting, rendezvous of \emph{identical} agents is possible only if the graph is asymmetric or the agents are placed in asymmetric positions on the graph. There exists a characterization of such instances (graphs and initial positions of agents) where rendezvous is solvable \cite{YKsolvable}.
If the agents are asynchronous, they can take advantage of the asymmetry in their initial positions by marking their initial position by a pebble~\cite{BasG01,kkmbook}.
On the other hand, if the agents have distinct labels then rendezvous is possible in any graph and any starting positions, without the need to mark nodes. 
The first deterministic synchronous rendezvous algorithm for agents with distinct labels was presented by Dessmark et al. \cite{DessmarkFKP06} and the time complexity of the algorithm was $O(n^5 \sqrt{\tau \log l} \log{n} + n^{10} \log^{2}n\log{l})$ for a graph of $n$ nodes, where $\tau$ is the delay in the starting time of the agents. Subsequent studies~\cite{KowalskiM08,TSZ07} improved this result and removed the parameter $\tau$ from the time complexity allowing for rendezvous in time polynomial in both $n$ and $\log{l}$. For the asynchronous case, De Marco et al. \cite{MarGKKPV06} provided an algorithm for rendezvous with a cost of $O(D\cdot \log{l})$ rounds, when the graph is known. For unknown arbitrary graphs, Czyżowicz et al. \cite{LabourelP10} gave the first algorithm for asynchronous rendezvous but the cost of this algorithm is at least exponential in the distance $D$ and the   degree $\Delta$ of the graph. Recently, Dieudonn\'{e} et al. \cite{DieudonnePV2013} provided an improvement over this result achieving asynchronous rendezvous in time polynomial in $n$ and $\log{l}$.
Rendezvous of agents starting from a finite distance $D$ in an infinite graph has been studied for the special cases when the graph is a line~\cite{MarGKKPV06,Stachowiak} or a grid~\cite{BampasGrid2010}, assuming that the agents have a sense of orientation and they know their own location in the labelled grid.

There have been several studies on the minimum capabilities needed by the agents to solve rendezvous. For example, the minimum memory required by an agent to solve rendezvous is known to be $\Theta(\log{n})$ for arbitrary graphs. Czyżowicz et al. \cite{logspaceRV} have provided a memory optimal algorithm for rendezvous, and there are studies on the tradeoff between time and space requirements for rendezvous~\cite{CKP-tradeoffs}.
In some papers, additional capacities are assumed for the agents to overcome other limitations, e.g. global vision is assumed to overcome memory limitations~\cite{kmp-oblivious} or the capability to mark nodes using tokens~\cite{kkmbook} or whiteboards~\cite{disc2010} is often used to break symmetry.
The model used in this paper can be seen as a special case of the oracle model for computation~\cite{oracle} where the agent is allowed to query an oracle that has global knowledge of the environment. However in our case, since the only queries are distance queries, the oracle can be implemented without complete knowledge of the graph topology.

\section{Model and Notations}

We model the environment as an undirected connected (possibly infinite) graph $G(V,E)$. The nodes of $V$ are unlabelled such that vertices of the same degree look identical to any agent (i.e. the nodes are anonymous). At each node of the graph, the edges incident to it are locally labelled, so that an agent arriving at a node can distinguish among them\footnote{Such an assumption is necessary to allow the agent to navigate in the graph.}. We assume that edges incident to a node $v$ are labelled by distinct integers (called port numbers) from the set $\{1,2,\dots, d(v)\}$, where $\degree{v}$ is the degree of node $v$. The degree of each node is finite and bounded by the parameter $\Delta$ (which is unknown to the agent). For any two distinct vertices $u,v \in V$, the distance between them, denoted by $\dist{u}{v}$, is the number of edges in any shortest path from $u$ to $v$ in $G$.

There are exactly two agents $a_1$ and $a_2$, and each agent $a_i$ has a distinct label $\labelOfAgent[\agent_i]$ $\in$ $\{0,1,\dots, L-1\}$ for some integer $L\geq 1$. An agent knows its own label but not that of the other agent. The agents have no prior knowledge of the graph.
Each agent starts from a distinct node of the graph and moves along the edges of the graph in synchronous steps following a deterministic algorithm. In other words, time is discretized into regular intervals called rounds; in each round, an agent at a node $v$ can either move to an adjacent node of $G$ or remain stationary at $v$. If the agent moves to an adjacent node $w$, the agent becomes aware of the port number of the edge through which it entered $w$.
The agent has no means of marking a node that it visits and the agents cannot communicate with each other. An agent can see the other agent only when both agents are on the same node (in particular the agents do not see each other if they cross on the same edge from opposite directions). The two agents start the algorithm in the same round (called round $0$) and rendezvous is achieved in the earliest round $T$ when the two agents are at the same node. We denote by $D$ the distance between the starting locations of the agents.

Contrary to previous studies on rendezvous, we assume that the agent is equipped with a device that measures the distance to the other agent. An agent at a node $v$ in round $t$ can make a query to this device (modelled as a function call \currentDist()) which returns the value $\dist{v}{u}$, where node $u$ is the location of the other agent in this round. In each round the agent can make one call to \currentDist() and depending on the value returned, the value of the agent's label, the current state of the agent, and the degree of the current node, the agent chooses a number between $0$ and $\degree{v}$ and leaves the current node $v$ through this port. We assume that the port number $0$ corresponds to a self-loop at node $v$ and if the agent chooses $0$ it remains at the same node $v$.
In this paper we do not restrict the memory of an agent in any way. Thus, the agent can memorize its complete history of moves up to the current round and store this as its internal state.


\section{Lower Bound for Distance-aware Rendezvous}

In this section we provide lower bounds on the rendezvous time for distance-aware agents. Observe that a trivial lower bound is $\Omega(D)$ since at least one of the two agents must traverse $D/2$ edges to achieve rendezvous. We could easily obtain a better lower bound for graphs of degree $\Delta$. The result below is folklore and we include it only for completeness.

\begin{lemma}
Rendezvous of two agents that are initially at a distance of $D$ apart in an unknown graph of maximum degree $\Delta$ requires $\Omega(D \cdot \Delta)$ rounds in the worst case.
\end{lemma}

\begin{proof}
Consider the caterpillar graph shown in Figure~\ref{fig:lowerbound1} obtained by taking a line of $D+1$ nodes and replacing each node with a star of size $\Delta-1$. Suppose that the two agents start at the endpoints of the path of length $D$ in this graph, as shown. If an agent traverses an edge leading to one of the leaves, it has no other option than to return and try another port. For any deterministic algorithm, an adversary can assign the port numbers in such a way that the agent needs to traverse all the $\Delta$ incident edges before it gets any closer to the other agent. Thus after $2 \Delta-1$ rounds the agents can reduce the distance between them by $2$. Using a repetition of this argument, the result follows.
\end{proof}

\begin{figure}[htb]
\includegraphics[width=\textwidth]{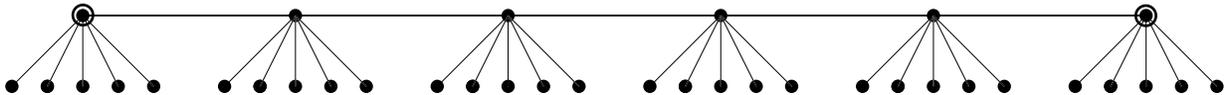}
\caption{ \small{The two agents are located at the specially marked nodes at the two ends of a Caterpillar graph.}
} \label{fig:lowerbound1}
\end{figure}

The above lower bound is independent of the agent's labels. We now present a bound which includes the parameter $L$ which bounds the size of the set of possible agent labels. It is known that rendezvous of two agents with distinct labels requires $\Omega(\log{L})$ rounds~\cite{DessmarkFKP06}. However this lower bound is for the simplest graph consisting of two nodes and a single edge connecting them. We provide below a better lower bound for arbitrary graphs of maximum degree $\Delta$.

\begin{figure}[tb]
\includegraphics[width=\textwidth]{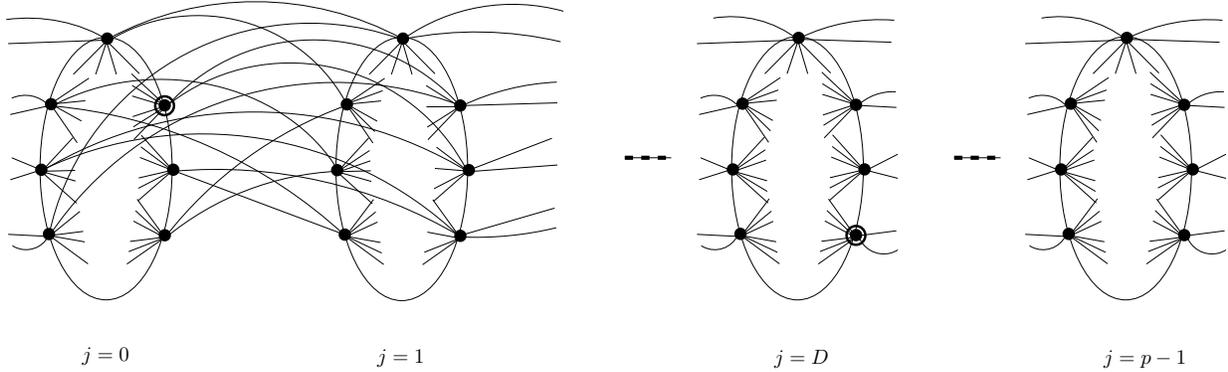}
\caption{ \small{A {$k$-clique-$p$-butterfly} graph where the two agents are located at the marked nodes at a distance $D>\log{k}$.}
} \label{fig:lowerbound2}
\end{figure}

Our method is constructive, i.e., given a deterministic rendezvous algorithm $\mathcal{A}$, we describe a graph and a procedure for the adversary to label the ports in a graph and to choose starting positions that will keep the agents executing $\mathcal{A}$ from meeting for the desired time period.
\begin{definition}
A \emph{$k$-clique-$p$-butterfly} is a $(k+3)$-regular graph on $(k \times p)$ vertices (denoted by $v_{0,0}$, $\ldots$, $v_{k-1,p-1}$), constructed as follows:
\begin{itemize}
\item  for each $0 \le j < p$, we connect all pairs of vertices from $v_{0,j},v_{1,j},\ldots,v_{k-1,j}$ (so they form a $k$-vertex clique, named $j$-th clique),
\item for any given $0 \le j,j' < p$, $0 \le i,i' < k$, we connect $v_{i,j}$ and $v_{i',j'}$ if $j' = (j+1) \bmod p$ and $i' \in \{2i \bmod k,(2i+1) \bmod k\}$.
\end{itemize}
\end{definition}

For two vertices, $v$ from $j$-th clique and $v'$ from $j'$-th clique, we say that their horizontal distance is $\min( (j-j') \bmod p, (j'-j) \bmod p)$. By the properties of the butterfly-type interconnections between the cliques, we observe that, if for two vertices the horizontal distance is at least $\log k$, then their actual distance in the graph is equal to the horizontal distance. See Figure~\ref{fig:lowerbound2} for an example of a \emph{$k$-clique-$p$-butterfly} graph where the agents are at a distance $D>\log{k}$.

\begin{theorem}
\label{th:lowerbound}
For any odd $k \ge 3$ and $p \ge 2\cdot\lfloor \log k \rfloor $, given a $k$-clique-$p$-butterfly (denoted further as $G$), for any deterministic algorithm $\mathcal{A}$, and integers $L>0$ and $D \ge \log k = \log (\Delta-3)$, there exist labels $0 \le \labelOfAgent[1],\labelOfAgent[2] < L$ and a port numbering of $G$ such that if two agents with these labels start at two vertices at distance $D$ in this graph, they will remain at distance $D$ for at least $\Theta(\Delta \log L / \log \Delta)$ steps.
\end{theorem}

\begin{proof}
The main idea of the proof is as follows. We will choose four special port numbers that will be assigned in every vertex to the ``bridge" edges connecting it to an adjacent clique (two ports to advance to the next clique and two ports to go backward). Note that moving inside the clique doesn't change the distance to the other agent.
By carefully picking the agent labels, we will ensure that both agents will choose the forward ports at the same rounds and the backward ports also at the same rounds, thus they will maintain a constant distance between themselves, for a sufficient amount of time and during this time period, any queries to the distance oracle would be useless.

Before we proceed to prove the theorem, we need to fix some notation and provide some basic lemmas.
We will consider the {$k$-clique-$p$-butterfly} graphs with a special type of port numbering where the ports on the two ends of each edge always form one of the pairs of values: $\{1,\Delta\}, \{2,\Delta-1\}, \ldots, \{\Delta/2,\Delta/2+1\}$. In other words, whenever an agent chooses to leave a node by port $j$, it arrives at the adjacent node by port $\Delta+1-j$. In a $\Delta$-regular graph with such a port numbering, an agent can never distinguish between any two vertices and it can learn nothing new about the graph by traversing it.
Thus, the algorithm $\mathcal{A}$ must choose a predefined sequence of ports to follow during the first $t>0$ rounds when the distance between the agents remains $D$.
For every $0 \le \labelOfAgent < L$ we denote the sequence of ports chosen by algorithm $\mathcal{A}$ (with input label $\labelOfAgent$) by $(\mathcal{P}_{\labelOfAgent}(i))_{i=0}^{i<t}$.

\begin{lemma}
\label{lem:low1}
There exist $p_1,p_2$, $0 < p_1,p_2 \le \Delta/2$, and a set $\mathcal{L} \subset \{0,1,\ldots,L-1\}$, $|\mathcal{L}| \ge L/2$, such that for each $\labelOfAgent \in \mathcal{L}$, at most $8 t /\Delta$ elements of $(\mathcal{P}_{\labelOfAgent}(i))_{i=0}^{i<t}$ are equal to one of the four values $\{p_1,p_2,\Delta+1-p_1,\Delta+1-p_2\}$ (in total).
\end{lemma}

\begin{proof}
In $G$, the port numbers on the edges are paired together as follows: $\{1,\Delta\}, \ldots, \{\Delta/2,\Delta/2+1\}$.
For each pair $\{p,\Delta+1-p\}$, let $S_p$ be the number of pairs $(i,\labelOfAgent)$, $0 \le i < t,0 \le \labelOfAgent < L$, such that  $\mathcal{P}_{\labelOfAgent}(i) \in \{p,\Delta+1-p\} $. Since $S_1+S_2+\cdots+S_{\Delta/2}\le t\cdot L$, there must exist $p_1,p_2$ such that $S_{p_1}+S_{p_2} \le t \cdot L \cdot 2 / (\Delta/2)$=${t \cdot L \cdot 4 / \Delta}$. Thus, there must exist a subset $\mathcal{L}$ containing at least half of the possible labels, such that for each $\labelOfAgent \in \mathcal{L}$, the number of $i$'s for which $\mathcal{P}_{\labelOfAgent}(i) \in \{p_1,p_2,\Delta+1-p_1,\Delta+1-p_2\}$ is not greater than $2 \cdot \frac{t \cdot L \cdot 4 / \Delta}{L} = 8t / \Delta$.
\end{proof}

Given the ports $p_1$ and $p_2$ as in the lemma above, the adversary can assign the pairs of port numbers $(p_1, \Delta+1-p_1)$ and $(p_2, \Delta+1-p_2)$ to those edges of $G$ that connect two adjacent cliques, such that ports $p_1$ and $p_2$ will take an agent forward to the next clique, ports $\Delta+1-p_1$ and $\Delta+1-p_2$ would take an agent backward to the previous clique and any other port will keep an agent within the same clique. We partition the set of port numbers into the subsets:
\begin{eqnarray*}
 A &=& \{p_1,p_2\},  \\
 B &=& \{\Delta+1-p_1,\Delta+1-p_2 \},  \\
 C &=& \{0,1,\dots \Delta\} \setminus (A \cup B).
\end{eqnarray*}
Informally, the ports in $A$ take an agent one step ahead, the ports in $B$ take it one step back, and the $C$ ports keep it in the same clique.

\begin{lemma}
\label{lem:low22}
There exist $0 \le \labelOfAgent[1]<\labelOfAgent[2] < L$ such that $\forall i \leq \lfloor \frac{\log L}{ 2 \log \Delta} \rfloor \cdot \frac\Delta 8$,
\begin{equation} \label{eq:A}
\mathcal{P}_{\labelOfAgent[1]}(i) \in A \text{ iff } \mathcal{P}_{\labelOfAgent[2]}(i) \in A, \text{ and}
\end{equation}
\begin{equation} \label{eq:B}
\mathcal{P}_{\labelOfAgent[1]}(i) \in B \text{ iff } \mathcal{P}_{\labelOfAgent[2]}(i) \in B.
\end{equation}
\end{lemma}
\begin{proof}
Let $t > 0$ be the first round such that for $i=t$, either \eqref{eq:A} or \eqref{eq:B} does not hold for any possible pair of labels $\labelOfAgent[1]$ and $\labelOfAgent[2]$ from the set $\{0,1,\ldots,L-1\}$. We will bound the value of $t$. Note that the sequence of ports $(\mathcal{P}_{\labelOfAgent}(i))_{i=0}^{i<t}$ can be written as a sequence of $A$,$B$ and $C$'s.
We choose the set of labels $\mathcal{L}$ from Lemma~\ref{lem:low1} and for any $\labelOfAgent \in \mathcal{L}$, let us count the number of possible sequences $(\mathcal{P}_{\labelOfAgent}(i))_{i=0}^{i<t}$ which correspond to distinct words from the alphabet $\{A, B, C\}$; let us denote this number as $X$.
Recall that any such sequence can have at most $8t/\Delta$ $A$-ports and $B$-ports in total (cf. Lemma~\ref{lem:low1}).
Assuming for the sake of notation that $T=8t/\Delta$ is an integer and using Stirling's approximation, we get:
$$X = \sum_{i=0}^{\lfloor T\rfloor} \binom{t}{i} \cdot 2^i \le 1+T \cdot \binom{t}{T} \cdot 2^{T}\le 1+T \cdot\frac{(2t)^{T}}{(8t/\Delta)!}$$
$$ \le 1+T\cdot\left(\frac{2t}{8t/\Delta\cdot\frac1e}\right)^{T} = 1 + T\left(\frac{e}{4}\Delta\right)^{T}. $$
Since for any two labels in $\mathcal{L}$ the algorithm chooses distinct sequences, we have $L/2 \leq |\mathcal{L}| \leq X$, which gives us $T \geq \lfloor \frac{\log L}{ 2 \log \Delta} \rfloor$ and thus $t \geq \lfloor \frac{\log L}{ 2 \log \Delta} \rfloor \cdot \frac\Delta 8$.
\end{proof}

We now return to the proof of the Theorem~\ref{th:lowerbound}.
From Lemma~\ref{lem:low22}, we know that there exist two agents with labels $\labelOfAgent[1],\labelOfAgent[2]$ such that these agents, when executing algorithm $\mathcal{A}$, will use the ports going forward or backward simultaneously during the first $t=\lfloor \frac{\log L}{ 2 \log \Delta} \rfloor \cdot \frac\Delta 8$ rounds. Since the distance in the graph is independent of relative positions inside cliques (for $D \ge \log k$), the distance between agents is maintained to be $D$ for at least $t=\Theta(\Delta \log L / \log \Delta)$ steps.
\end{proof}

The results of this section provide a lower bound of $\Omega(\Delta (D + \log L / \log \Delta))$ rounds for rendezvous of distance-aware agents.

\section{Upper Bound and Algorithm for Rendezvous}

In this section we provide 
an algorithm that guarantees rendezvous of two distance-aware agents in $O(\Delta(D + \log l))$ rounds.
The algorithm is divided into four procedures.
For each procedure, first we provide some intuitions on its behavior, then we give its formal description in the form of pseudo-code.
After the descriptions of all procedures, the proof of correctness and analysis of rendezvous time of the algorithm are given.

\medskip
We start with a procedure $\procSingleTest$ (see Algorithm~\ref{proc:singleTest}), which attempts to decrease the distance between the two agents. The input consists of a positive integer $\delta$ and a bit $b\in\{0,1\}$.
Recall that the command $\currentDist$() performs the oracle query and provides the current distance from the other agent.
We also introduce a command $\moveRound{x}$, where $x$ is an integer, which behaves as follows.
If $x\in\{1,\ldots,\degree{v}\}$, where $v$ is the node currently occupied by the executing agent, then $\moveRound{x}$ forces the executing agent to move from $v$ by taking port $x$, and the value returned by the $\moveRound{x}$ command is the entry port at the arrival node.
Otherwise, that is when $x\notin\{1,\ldots,\degree{v}\}$, the executing agent stays idle in the given round, and $\moveRound{x}$ returns $0$.
(We pass $0$ as an argument to deliberately make an agent idle.)
\begin{algorithm} \caption{Procedure $\procSingleTest(\delta,b)$}
\label{proc:singleTest}
\begin{algorithmic}
\REQUIRE Two integers $\delta\geq 1$ and $b\in\{0,1\}$.
\ENSURE $\returnSuccess$ if the distance between agents decreases in some round; $\returnFailure$ otherwise.
  \FOR{$i \leftarrow 1$ \TO $\delta$}
     \STATE $x \leftarrow \currentDist$()
     \STATE $p \leftarrow \moveRound{i\cdot b}$
     \STATE $y \leftarrow \currentDist$()
     \IF{$y<x$}
        \RETURN $\returnSuccess$ \COMMENT{When the distance to the other agent decreased.}
     \ENDIF
     \STATE $\moveRound{p\cdot b}$  \COMMENT{Going back along the same edge.}
  \ENDFOR
  \RETURN $\returnFailure$ \COMMENT{When the distance to the other agent never decreased.}
\end{algorithmic}
\end{algorithm}

The interpretation of the input variable $b$ is that whenever $b=0$, then the agent does not perform any movements during the execution of $\procSingleTest$.
Note that if $b=1$, then $\moveRound{i\cdot b}$ forces an agent to move only if there exists an edge with port $i$ at $v$.

Suppose that, in the same round, both agents perform calls to procedure $\procSingleTest$ with input parameters $\delta_1,b_1$ and $\delta_2,b_2$, respectively. We will always ensure that $\delta=\delta_1=\delta_2$.
Informally, $b_1=b_2=1$ implies that both agents iteratively take ports $1,\ldots,\delta$ (skipping the ones not present at the current node), ending the process if the distance between them decreases.
Clearly, if $b_1=b_2=0$, then both agents just stay idle during the $2\delta$ rounds and the procedure returns $\returnFailure$.
If $b_1\neq b_2$, then one agent stays idle while the other `tests ports'.
If we ensure that $\delta$ exceeds the degree of the node occupied by the agent that performs the movements, then the procedure will return $\returnSuccess$ whenever $b_1\neq b_2$.

\medskip
We now describe a procedure $\procBoundDegrees$ with input variable $b\in\{0,1\}$ (see Algorithm~\ref{proc:BoundDegrees}).
Informally speaking, for each $2^l$ such that $2^l<\degree{v}$, where $v$ is the node occupied at the beginning of the execution of the procedure, the executing agent stays idle for $2^{l+1}$ rounds (this is achieved by the call to $\procSingleTest(2^l,0)$.
This part is independent of $b$.
Then, $\procSingleTest(2^{\lceil\log_2 \degree{v} \rceil},b)$ is called.
If $b=0$, then the agent stays idle for another $2^{\lceil\log_2 \degree{v} \rceil+1}$ rounds.
If $b=1$, then the agent sequentially explores all ports at $v$ (Note that the value of $2^{\lceil\log_2 \degree{v} \rceil}$ exceeds the degree of $v$).
The above process is interrupted whenever the agents observe that the distance between them decreased, and the procedure returns $\returnSuccess$ in that case.
\begin{algorithm} \caption{Procedure $\procBoundDegrees(b)$}
\label{proc:BoundDegrees}
\begin{algorithmic}
\REQUIRE An integer $b\in\{0,1\}$.
\ENSURE $\returnSuccess$ or $\returnFailure$.
  \STATE Let $v$ be the currently occupied node.
  \FOR{$l \leftarrow 0$ \TO $\lceil\log_2\degree{v}\rceil-1$}
     \STATE $s \leftarrow \procSingleTest(2^l,0)$
     \IF{$s=\returnSuccess$}
        \RETURN $\returnSuccess$
     \ENDIF
  \ENDFOR

  \STATE $s \leftarrow \procSingleTest(2^{\lceil\log_2\degree{v}\rceil},b)$
  \RETURN $s$
\end{algorithmic}
\end{algorithm}

The observation given below follows directly from the formulation of procedure $\procBoundDegrees$.
\begin{observation} \label{obs:same-node}
Suppose that agent $\agent_i$ occupies node $v_i$, $i\in\{1,2\}$, and executes procedure $\procBoundDegrees$ at the beginning of round $r$.
If procedure $\procSingleTest$ does not return $\returnSuccess$ in the first $j$ iterations of $\procBoundDegrees$, then:
\begin{enumerate}[label={\normalfont(\roman*)},leftmargin=*]
 \item\label{it:obs:1} $\agent_i$ occupies $v_i$ at the end of the $j$-th iteration of $\procBoundDegrees$,
 \item\label{it:obs:2} both agents end the execution of the $j$-th iteration in round $r+2^{j}-1$.
\end{enumerate}
\end{observation}

Denote $I_0=(0,1]$ and $I_j=(2^{j-1},2^{j}]$ for $j\geq 1$.
We say that two nodes $u$ and $v$ are \emph{similar} if there exists $j\geq 0$ such that $\degree{u}\in I_j$ and $\degree{v}\in I_j$.

\begin{lemma} \label{lem:similar-nodes}
Let $r$ be some integer.
Suppose that agent $\agent_i$ is present at $v_i$, $i\in\{1,2\}$, and calls in round $r$ procedure $\procBoundDegrees$ with input value $b_i$.
Then:
\begin{enumerate}[label={\normalfont(\roman*)},leftmargin=*]
 \item\label{it:sim:1} for $b_1=b_2=1$, if both calls to $\procBoundDegrees$ return $\returnFailure$, then the nodes $v_1$ and $v_2$ are similar, and
 \item\label{it:sim:2} if $b_1\neq b_2$ and $v_1$ and $v_2$ are similar, then both calls to $\procBoundDegrees$ return $\returnSuccess$.
\end{enumerate}
\end{lemma}
\begin{proof}
We first prove \ref{it:sim:2}.
Suppose without loss of generality that $b_1=1$ and $b_2=0$.
Since the nodes $v_1$ and $v_2$ are similar, $x=\lceil\log_2\degree{v_1}\rceil=\lceil\log_2\degree{v_2}\rceil$ and therefore the executions of procedure $\procBoundDegrees$ have the same number of iterations of the `for' loop.
If, in one of those iterations, the execution of procedure $\procBoundDegrees$ ends, then \ref{it:sim:2} holds and hence suppose that this is not the case.
Thus, the calls to $\procSingleTest(2^{x},b_1)$ and $\procSingleTest(2^{x},b_2)$ are made by the agents.
Moreover, by Observation~\ref{obs:same-node}, both calls are made in the same round and when agent $\agent_i$ is at $v_i$, $i\in\{1,2\}$.
During these calls, $\agent_2$ stays idle during $2^{x+1}$ rounds (because $b_2=0$) while $\agent_1$ explores all ports at $v_1$ during the same $2^{x+1}$ rounds (because $b_1=1$ and $2^{\lceil\log_2\degree{v_1}\rceil}\geq \degree{v_1}$).
This guarantees that the latter calls to $\procSingleTest$ return $\returnSuccess$, which completes the proof of \ref{it:sim:2}.

We now prove \ref{it:sim:1}.
Let $b_1=b_2=1$.
We argue that if the nodes $v_1$ and $v_2$ are not similar, then the call to $\procBoundDegrees$ results in returning $\returnSuccess$.
Suppose without loss of generality that $\degree{v_i}\in I_{j_i}$, $i\in\{1,2\}$, where $j_1<j_2$.
The number of iterations of the `for' loop of procedure $\procBoundDegrees$ executed by $\agent_i$ is $\lceil\log_2\degree{v_i}\rceil=j_i$, $i\in\{1,2\}$.
Thus, by Observation~\ref{obs:same-node}\ref{it:obs:1}, $\agent_i$ occupies $v_i$ at the end of $j_1$-th iteration of $\procBoundDegrees$ for each $i\in\{1,2\}$.
Moreover, by Observation~\ref{obs:same-node}\ref{it:obs:2}, after finishing the execution of the `for' loop, $\agent_1$ calls $\procSingleTest(2^{j_1},b_1)$ while $\agent_2$ calls $\procSingleTest(2^{j_1},0)$ in the $(j_1+1)$-st iteration of the `for' loop of $\procBoundDegrees$.
Also, both of the above-mentioned calls to $\procSingleTest$ are made in the same round $r'$.
By similar arguments as when proving \ref{it:sim:2}, we obtain that condition \ref{it:sim:1} holds.
\end{proof}

\medskip
Note that if procedure $\procBoundDegrees$ returns $\returnSuccess$, then the agents get closer, and we can repeat the same process for at their current locations.
However, for some nodes procedure $\procBoundDegrees$ may return $\returnFailure$ and then procedure $\procCompareLabels$ described below (see Algorithm~\ref{proc:CompareLabels}) helps to break the symmetry.
\begin{algorithm} \caption{Procedure $\procCompareLabels$}
\label{proc:CompareLabels}
\begin{algorithmic}
\REQUIRE None.
\ENSURE $\labelsBit{\labelExtended{\labelOfAgent}}{i}$ such that $i$ is distinguishing for the extended labels of the agents.
  \STATE Let $\labelExtended{\labelOfAgent}$ be the extended label of the executing agent.
  \FOR{$i \leftarrow 1$ \TO $\lceil\log_2\labelExtended{\labelOfAgent}\rceil$}
      \STATE $s \leftarrow \procBoundDegrees(\labelsBit{\labelExtended{\labelOfAgent}}{i})$
      \IF{$s=\returnSuccess$}
         \RETURN $\labelsBit{\labelExtended{\labelOfAgent}}{i}$
      \ENDIF
  \ENDFOR
\end{algorithmic}
\end{algorithm}
Procedure $\procCompareLabels$ uses the notion of extended labels.
The \emph{extended label} $\labelExtended{\labelOfAgent}$ of a label $\labelOfAgent$ is an integer whose $j$-th bit is defined as follows:
\begin{equation*}
\labelsBit{\labelExtended{\labelOfAgent}}{j}=
   \begin{cases}
        \labelsBit{\labelOfAgent}{\lceil j/2\rceil}, & \textup{for }j\in\{1,3,5,\ldots,2\lceil\log_2\labelOfAgent\rceil-1\}, \\
        1,                                           & \textup{for }j=2\lceil\log_2\labelOfAgent\rceil, \\
        0,                                           & \textup{otherwise}.

   \end{cases}
\end{equation*}
The index $j=2\lceil\log_2\labelOfAgent\rceil$ is called the \emph{terminating bit} of $\labelExtended{\labelOfAgent}$ (this is the last bit set to $1$).
Informally, the odd positions of the extended label are the the bits of $\labelOfAgent$ while the even positions are all zeros, except for the terminating bit.
We say that an index $i$ is \emph{distinguishing} for two extended labels $\labelExtended{\labelOfAgent}$ and $\labelExtended{\labelOfAgent'}$ if $\labelsBit{\labelExtended{\labelOfAgent}}{i}\neq\labelsBit{\labelExtended{\labelOfAgent'}}{i}$.

Informally, procedure $\procCompareLabels$ iterates over the bits of the extended label of the executing agent in order to find an index $i$ that is distinguishing for the extended labels of the two agents.
The construction of extended labels guarantees that there exists a distinguishing index $i$ not greater than the smaller length of the two extended labels and hence $\procBoundDegrees$ returns $\returnSuccess$ at the latest in the $i$-th iteration of the `for' loop of procedure $\procCompareLabels$ (the formal proof is given later; we remark here that if we used the label instead of the extended label, then the number of iterations of the `for loop' would have to be equal to the length of the greater label to ensure rendezvous).

We postpone the analysis of procedure $\procCompareLabels$ (given in Lemma~\ref{lem:CompLabels}) as it depends on the context at which it is called by the main procedure.

\medskip
We finally describe the main procedure $\procRV$ (see Algorithm~\ref{proc:RV}).
\begin{algorithm} \caption{Procedure $\procRV$}
\label{proc:RV}
\begin{algorithmic}
  \STATE $s \leftarrow \returnSuccess$
  \WHILE{rendezvous not achieved \AND $s=\returnSuccess$}
     \STATE\label{line:RV:b1} $s \leftarrow \procBoundDegrees(1)$
  \ENDWHILE
  \STATE $b \leftarrow \procCompareLabels$
  \WHILE{rendezvous not achieved}
     \STATE $\procBoundDegrees(b)$
  \ENDWHILE
\end{algorithmic}
\end{algorithm}
We start with its intuitive description.
The first `while' loop iteratively calls procedure $\procBoundDegrees(1)$ as long as its execution gets the agents closer to each other.
If a call to $\procBoundDegrees(1)$ does not achieve that, then (as we formally prove later) the agents observed the same distance between each other while both explored all ports at their respective locations.
This is significant as both agents learn that they occupy nodes whose degrees are in the same interval $I_j$ for some $j\geq 0$.
In other words, the agents learn an asymptotically tight upper bound on both degrees.
Then, procedure $\procCompareLabels$ is called and uses the above fact as well as the labels of the agents to break the symmetry that occurs at the current agents' nodes.
Note that $\procCompareLabels$ returns either $0$ or $1$ and in this case different values are returned for both agents (see Lemma~\ref{lem:CompLabels} below).
Thus, the agent whose execution of $\procCompareLabels$ returned $0$ stays idle from now on.
The other agents continues making calls to $\procBoundDegrees(1)$ and since each execution results in exploring all ports at the currently occupied node, each execution gets the agent one step closer to the one that is idle.
We also remark that the respective calls to procedure $\procBoundDegrees$ in the second `while' loop of $\procCompareLabels$ are not necessarily `synced', that is, the $j$-th of those calls can be made in different rounds by the agents.
This, however, is not important as one of the agents stays idle and the other one performs appropriate movements.

\medskip
Lemma~\ref{lem:CompLabels} analyzes the only call to procedure $\procCompareLabels$ made by procedure $\procRV$.
Then, Theorem~\ref{thm:upper-bound} provides the upper bound on the rendezvous time for distance-aware agents in arbitrary networks.

\begin{lemma} \label{lem:CompLabels}
Whenever procedure $\procCompareLabels$ is called by both agents during the execution of procedure $\procRV$,
both agents finish the execution of $\procCompareLabels$ in the same round and the values returned by $\procCompareLabels$ are different for the two agents.
\end{lemma}
\begin{proof}
Since both agents call $\procCompareLabels$, none of the preceding calls to procedure $\procBoundDegrees$ returns $\returnSuccess$ and the agents do not rendezvous prior to the call to $\procCompareLabels$.
By Observation~\ref{obs:same-node} and a simple inductive argument, the calls to procedure $\procCompareLabels$ made by both agents end in the same round $r$.

It remains to prove that the calls to $\procCompareLabels$ return different values.
By Lemma~\ref{lem:similar-nodes}\ref{it:sim:1}, at the beginning of round $r$ the agents $\agent_1$ and $\agent_2$ are, respectively, at two similar nodes $v_1$ and $v_2$.
Let $j\geq 0$ be the minimum distinguishing index for agents' labels, i.e., $\labelsBit{\labelExtended{\labelOfAgent[\agent_1]}}{j}\neq\labelsBit{\labelExtended{\labelOfAgent[\agent_2]}}{j}$ and $\labelsBit{\labelExtended{\labelOfAgent[\agent_1]}}{j'}\neq\labelsBit{\labelExtended{\labelOfAgent[\agent_2]}}{j'}$ for each $1\leq j'<j$.
Such an index $j$ exists because the labels of the agents are different and hence the extended labels have a distinguishing index.
Moreover,
\[j\leq\min\left\{ \lceil\log_2\labelExtended{\labelOfAgent[\agent_1]}\rceil, \lceil\log_2\labelExtended{\labelOfAgent[\agent_2]}\rceil \right\}.\]
Indeed, if the two labels are of the same length, then the extended labels are of the same length.
If, on the other hand, the labels have different lengths, then the terminating bits are at different positions, which in particular implies that the terminating bit of the smaller label is at position that is distinguishing for the two extended labels.

By assumption, Observation~\ref{obs:same-node} and an inductive argument, $\agent_i$ is at $v_i$ at the beginning of the $j'$-th call to procedure $\procBoundDegrees$ during the execution of procedure $\procCompareLabels$, where $j'\leq j$.
The last execution of procedure $\procBoundDegrees$ preceding the call to $\procCompareLabels$ returns $\returnFailure$.
Hence, if $\labelsBit{\labelExtended{\labelOfAgent[\agent_1]}}{j'}=\labelsBit{\labelExtended{\labelOfAgent[\agent_2]}}{j'}=1$, then the $j'$-th call to $\procBoundDegrees$ returns $\returnFailure$.
Clearly, if $\labelsBit{\labelExtended{\labelOfAgent[\agent_1]}}{j'}=\labelsBit{\labelExtended{\labelOfAgent[\agent_2]}}{j'}=0$, then the agents stay idle during the $j'$-th call to procedure $\procBoundDegrees$ which also implies that it returns $\returnFailure$.

Thus, the above proves that the $j$-th call to procedure $\procBoundDegrees$ takes place during the execution of procedure $\procCompareLabels$ and, by Lemma~\ref{lem:similar-nodes}\ref{it:sim:2}, it returns $\returnSuccess$ for both agents.
Thus, procedure $\procCompareLabels$ returns the respective bits of the extended label at position that is distinguishing, which completes the proof.
\end{proof}

\begin{theorem} \label{thm:upper-bound}
Suppose that agent $\agent_i$ with label $\labelOfAgent[\agent_i]$ initially occupies node $v_i$, $i\in\{1,2\}$.
Procedure $\procRV$ guarantees that $\agent_1$ and $\agent_2$ rendezvous within
$O(\Delta\cdot (D + \min_i\{\log\labelOfAgent[\agent_i] \}))$ rounds where $D=\dist{v_1}{v_2}$.
\end{theorem}

\begin{proof}
We first prove that the execution of procedure $\procRV$ guarantees rendezvous.
If the agents rendezvous during the execution of the first `for' loop, then the claim follows and hence suppose that this is not the case.
Denote by $b_i$ the value of the variable $b$ returned by the call to procedure $\procCompareLabels$ by agent $\agent_i$, $i\in\{1,2\}$.
By Lemma~\ref{lem:CompLabels}, $b_1\neq b_2$.
Let without loss of generality, $b_1=0$ and $b_2=1$.
This implies that the agent $\agent_1$ stays idle indefinitely.
The agent $\agent_2$, during execution of procedure $\procBoundDegrees(b_2)$ called in the second `for' loop of procedure $\procRV$, explores all ports of the currently occupied node.
Thus, the distance between agents decreases in some round which implies that each such call to procedure $\procBoundDegrees(b_2)$ returns $\returnSuccess$.
Thus, the agents rendezvous eventually.

Now we analyze the rendezvous time.
Each call to $\procBoundDegrees$ takes $O(\Delta)$ rounds.
Moreover, each such call, except for at most one, made directly by procedure $\procRV$ ensures that the distance between the agents decreases.
This follows immediately for calls to $\procBoundDegrees$ preceding the call to procedure $\procCompareLabels$ since those calls, possibly except for the last one, return $\returnSuccess$.
As for the remaining calls to $\procBoundDegrees$, the above claim is due to the fact that the input values are different for both agents due to Lemma~\ref{lem:CompLabels}.
Thus, the total number of rounds due to all calls to $\procBoundDegrees$ made directly by procedure $\procRV$ is $O(\Delta\cdot\dist{v_1}{v_2})$.
The number of iterations of the `for' loop of procedure $\procCompareLabels$ is $O(\min\{\log\labelOfAgent[\agent_1],\log\labelOfAgent[\agent_2]\})$, each resulting in $O(\Delta)$ rounds (the call to $\procBoundDegrees$).
\end{proof}

\section{Conclusions}

This paper presented a new model for mobile agent computation by providing the agents with the capability of measuring distances to each other (or  detecting changes in distances) at each step. We show that this simple mechanism allows us to reduce the time to rendezvous from exponential to polynomial in the degree of the graph. Assuming that such a distance measuring device is available to the agents, one could ask what other problems can be solved more easily using this additional capability. For example, the agents could use this mechanism for communication at distance by moving back and forth, when there are no other means of communication. This opens up a new area of research which is worth investigating.


\end{document}